\documentclass[12pt,a4paper]{article}
\usepackage{a4}
\usepackage[a4paper,left=3.4cm,right=3.4cm,top=3.9cm,bottom=3.2cm]{geometry}
\usepackage{amssymb}
\usepackage{amsmath}
\usepackage{amsthm}
\usepackage{graphicx}
\usepackage[utf8]{inputenc}
\usepackage[T1]{fontenc}
\usepackage[english]{babel}
\usepackage{placeins}
\newcounter{cptRef}
\newcounter{cptTh}
\newtheorem{theorem}[cptTh]{Theorem}

\newtheorem{dfn}[cptTh]{Definition}
\newtheorem{obs}[cptTh]{Observation}
\newtheorem{notation}[cptTh]{Notation}

\usepackage[linesnumbered,longend,noline]{algorithm2e}
\SetKwProg{Fn}{Function}{}{end}
\DontPrintSemicolon

\usepackage{color}
\usepackage[colorlinks=false]{hyperref}
\hypersetup{pdftitle=Group connectivity: Z\_4 v. Z\_2\^{}2}
\hypersetup{pdfauthor={Radek Hušek, Lucie Mohelníková, Robert Šámal}}

\def\Z{{\mathbb Z}}

\def\set#1{\left\{#1\right\}}

\def\abs#1{\left|#1\right|}

\newcommand{\ie}{i.\,e.\ }

\def\NULL{\ensuremath{\text{\tt NULL}}}

\begin{document}

\title{Group Connectivity: $\Z_4$ v.~$\Z_2^2$}

\author{%
  Radek Hušek\and
  Lucie Mohelníková\and
  Robert Šámal\thanks{
All authors are members of Computer Science Institute of Charles University,
Prague, Czech Republic.
Email: {\tt \{husek,samal\}@iuuk.mff.cuni.cz}
}}

\date{}

\maketitle

\begin{abstract}
We answer a question on group connectivity suggested by Jaeger et al.~[Group connectivity of
graphs -- A nonhomogeneous analogue of nowhere-zero flow properties, JCTB 1992]: 
we find that $\Z_2^2$-connectivity does not imply $\Z_4$-connectivity, neither vice versa. 
We use a computer to find the graphs certifying this and to verify their properties
using non-trivial enumerative algorithm.
While the graphs are small (the largest has 15 vertices and 21~edges),
a computer-free approach remains elusive.
\end{abstract}

\section{Introduction}

A \emph{flow} in a digraph~$G=(V,E)$ is an assignment of values of some abelian group~$\Gamma$ to 
edges of~$G$ such that Kirchhoff's law is valid at every vertex. 
Formally, $\varphi\colon E \to \Gamma$ satisfies 
$$
   \sum_{\hbox{\footnotesize $e$ ends at~$v$}} \varphi(e) = \sum_{\hbox{\footnotesize $e$ starts at~$v$}} \varphi(e)
$$
for every vertex~$v \in V$. 
We say a flow is \emph{nowhere-zero} if it does not use value~$0$ at any edge. 

Tutte~\cite{Tutte} started the study of nowhere-zero flows by observing, that 
a plane digraph~$G$ has a nowhere-zero flow in~$\Z_k$ if and only if its plane dual~$G^*$ is $k$-colorable
(we do not consider orientation of the edges for the coloring). 
This motivated several famous conjectures, we mention just the 5-flow conjecture (due to Tutte): 
every bridgeless graph has a nowhere-zero flow in~$\Z_5$. 
A motivating feature of the theory of nowhere-zero flows are several nice properties, starting 
with the ones discovered by Tutte. In particular: 

\begin{theorem}[Tutte~'54~\cite{Tutte}] \label{Tuttegroup} 
Let $\Gamma$ be an abelian group with $k$-elements. Then for every digraph 
the existence of a nowhere-zero $\Gamma$-flow is equivalent with the existence of 
a nowhere-zero $\Z_k$-flow. 
\end{theorem}

\begin{theorem}[Tutte~'54~\cite{Tutte}] \label{Tutteint} 
The existence of $\Z_k$-flow is equivalent with the existence of a nowhere-zero integer flow, that uses only values 
$\pm 1$, $\pm 2$, \dots, $\pm (k-1)$. 
\end{theorem}

Jaeger et~al.~\cite{JLPT} introduced a variant of nowhere-zero flows called \emph{group connectivity}. 
A digraph~$G = (V,E)$ is \emph{$\Gamma$-connected} if for every mapping $h \colon E \to \Gamma$
there is a $\Gamma$-flow~$\varphi$ on~$G$ that satisfies $\varphi(e) \ne h(e)$ for every edge~$e \in E$. 
As we may choose the ``forbidden values'' $h \equiv 0$, every $\Gamma$-connected digraph has a nowhere-zero $\Gamma$-flow; 
the converse is false, however. While the notion of group connectivity is stronger than 
the existence of nowhere-zero flows, it is also more versatile, in particular the notion lends itself 
more easily to proofs by induction. This is a consequence of an alternative definition of group connectivity: 
instead of looking for a flow, we may check existence of a mapping $E \to \Gamma$ that has prespecified 
surplus at each vertex.

It is easy to see that both the existence of a nowhere-zero $\Gamma$-flow
and $\Gamma$-connectivity do not change when we reverse the orientation of an edge of the digraph 
(we only need to change the corresponding flow value from~$x$ to $-x$). 
Thus, we will say that an undirected graph~$G$ has a nowhere-zero $\Gamma$-flow (is $\Gamma$-connected) 
if some (equivalently every) orientation of~$G$ has a nowhere-zero $\Gamma$-flow (is $\Gamma$-connected).
Also, using the definition of group connectivity working with vertex surpluses, we observe that
group connectivity is monotone with respect to edge addition -- if $G$ is $\Gamma$-connected then
$G + e$ is $\Gamma$-connected for any edge $e$.

Some results on nowhere-zero flows extend to the stronger notion of group connectivity. 
A celebrated recent example of this is the solution to the Jaeger's
conjecture by Thomassen et al.~\cite{LTWZ}, but there are many more. 
Thus, it is worthwhile to understand the properties of group connectivity in more detail.

\begin{figure}[h]
  \centering
  \includegraphics{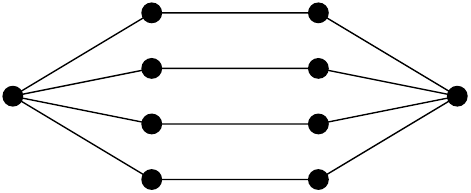}
  \caption{A graph which is $\Z_5$ but not $\Z_6$-connected\label{fig:5y6n}}
\end{figure}

However, some nice properties of group-valued flows are not shared by group connectivity. 
In particular Jaeger~\cite{JLPT} showed that there is a graph (Figure~\ref{fig:5y6n})
that is $\Z_5$-connected, but not $\Z_6$-connected. 
This contrasts with the situation for flows: Suppose $G$ has a nowhere-zero flow in $\Z_5$ but not 
in~$\Z_6$. Using Theorem~\ref{Tutteint} twice, we find that $G$ has a nowhere-zero integer flow with values bounded in 
absolute 
value by~$5$, but not one bounded by~$6$, a clear contradiction. 

An analogy of Theorem~\ref{Tuttegroup} is more subtle. Indeed, in Section~3.1 of~\cite{JLPT} the authors mention: 
``\dots\ we do not know of any $\Z_4$-connected graph which is not $\Z_2\times\Z_2$-connected, or vice versa. 
Neither can we prove that such graphs do not exist.'' 
Our main result is the resolution to this natural question. 

\begin{theorem}\label{main}~
\begin{enumerate}
\item 
  There is a graph that is $\Z_2^2$-connected but not $\Z_4$-connected. 
\item 
  There is a graph that is $\Z_4$-connected but not $\Z_2^2$-connected. 
\end{enumerate}
\end{theorem}

Because our result is computer aided, we do not present proof in classical sense.
Instead we present overview of our approach and examples of graphs proving Theorem~\ref{main}
in the next section. In Section~\ref{sec:algo} we describe the algorithm we used to test group
connectivity, and we add some implementation notes in Section~\ref{sec:impl}.

\section{The Group Connectivity Conjecture and Results}

When looking for graphs certifying Theorem~\ref{main}, we only need to consider 
graphs that do have nowhere-zero $\Z_2^2$-flow (equivalently, by Theorem~\ref{Tuttegroup},
nowhere-zero $\Z_4$-flow). It is natural to examine cubic graphs (and their subdivisions) due to
the following theorem:

\begin{theorem}[Jaeger et al. \cite{JLPT}]
Let $G$ be an 4-edge-connected graph. Then $G$ is both $\Z_2^2$- and $\Z_4$-connected.
\end{theorem}

In contrary to the usual case, however, we
are not interested in snarks (cubic graphs that fail to be edge 3-colorable), as those 
do not have nowhere-zero $\Z_2^2$-flow.

We note that subdividing an edge has no effect on the existence of a nowhere-zero flow 
(the new edge can have the same flow value as before). It makes the group connectivity 
stronger -- in effect, we are forbidding one more value on an edge. 
This suggests the following strategy: 
\begin{enumerate}
  \item pick an arbitrary\,/\,random 3-regular graph and
  \item repeatedly subdivide an edge and check $\Z_2^2$- and $\Z_4$-connectivity.
\end{enumerate}

\begin{figure}[h]
\begin{center}
\includegraphics[width=6cm]{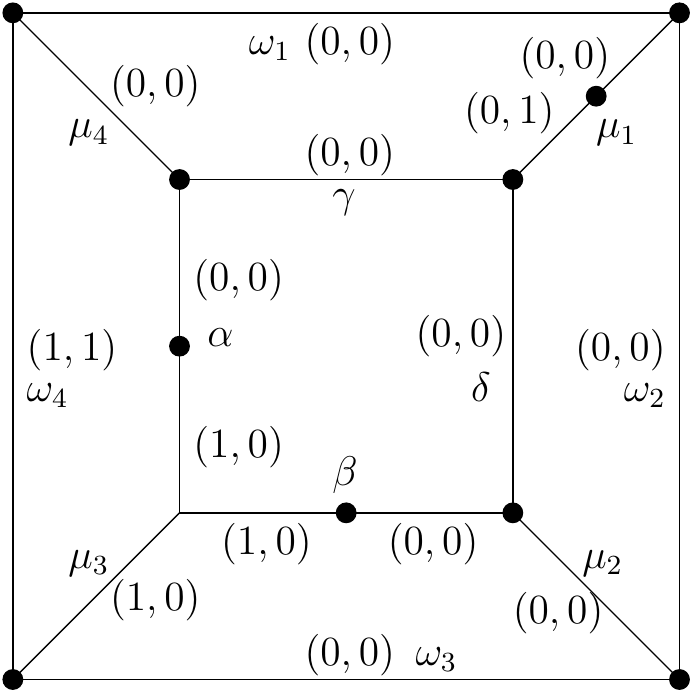}
\caption{A subdivision of cube which is $\Z_4$-connected but not $\Z_2^2$-connected with
forbidden assignment for which no satisfying $\Z_2^2$-flow exists and names for hypothetical flow values.}
\label{fig:cube_sub}
\label{fig:cube_sub_forb}
\end{center}
\end{figure}

This procedure yielded the graph in Figure~\ref{fig:cube_sub},
which appeared in the master thesis of the second author~\cite{Lysi}.
This graph is $\Z_4$- but not $\Z_2^2$-connected. 
Later, with more effective implementation (see the next section) by the first author,
we found graphs that are $\Z_2^2$- but not $\Z_4$-connected. 
The smallest among them are (threefold) subdivisions of cubic graphs on 12 vertices
(for an example see Figure~\ref{fig:graphs}).
We also include a proof that graph in Figure~\ref{fig:cube_sub} is not $\Z_2^2$-connected
which is not computer-aided:

\begin{theorem} \label{thm:non-conn}
The subdivision of cube shown in Figure~\ref{fig:cube_sub} is not $\Z_2^2$-connected.
\end{theorem}

\begin{figure}[h]
\begin{center}
\includegraphics[width=6cm]{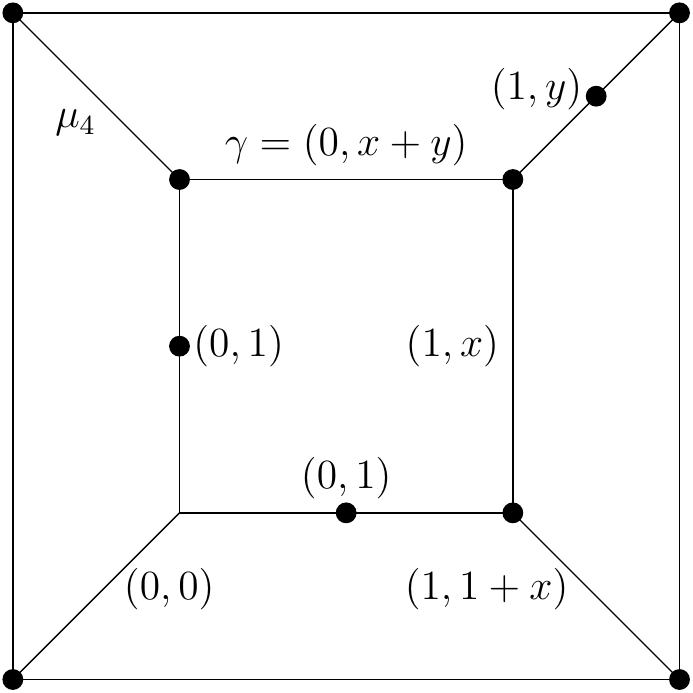}
\hskip 0pt plus 2fil
\includegraphics[width=6cm]{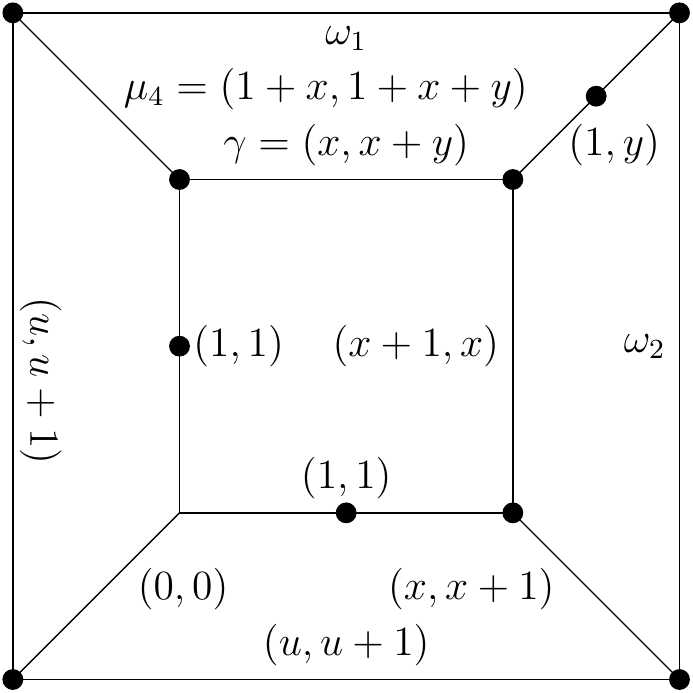}
\caption{Cases $\alpha = (0, 1)$ and $\alpha = (1,1)$ with fragments of hypothetical flows.}
\label{fig:cube_flows}
\end{center}
\end{figure}

\begin{proof}
We will show that for the assignment of the forbidden values in Figure~\ref{fig:cube_sub_forb} there exists
no satisfying $\Z_2^2$-flow. First observe that values $\alpha$ and $\beta$ are of form $(.,1)$
which implies $\mu_3 = (.,0)$. So $\mu_3$ is always $(0,0)$ and $\alpha = \beta$. Also $\mu_1 = (1, .)$.

Propagation of values of flow in the case $\alpha = (0,1)$ is shown in Figure~\ref{fig:cube_flows}, on the left.
As $\mu_2 \neq (0,0)$, we have $\delta = (1, .)$.
The value
$x+y$ is 1 because $\gamma$ is forbidden to be $(0,0)$ but this forces $\mu_4 = (0,0)$ which is also forbidden.
In the case $\alpha=(1,1)$ (Figure~\ref{fig:cube_flows}, on the right), 
we again combine the forbidden values to give possible form for $\mu_2$ and $\delta$, and also $\omega_3$,
$\omega_4$. In particular $\omega_3 = \omega_4 \not\in \set{ (0,0), (1,1)}$, so we may write
$\omega_3 = (u, u+1)$.
The edge $\gamma$ forbids $x = y = 0$ and the edge $\mu_4$
forbids $x=1$, $y=0$, so $y = 1$ and $\mu_4 = (x+1, x)$. So either $\omega_1 = (u+x+1, u+x+1)$
or $\omega_2 = (u + x, u + x)$ are $(0,0)$. Hence no satisfying flow exists.
\end{proof}

\begin{figure}[h]
  \begin{center}
    \begin{minipage}{.45\textwidth}
      \center
      \small $\Z_4$: NO \hfil $\Z_2^2$: YES \\
      \includegraphics[width=.95\textwidth]{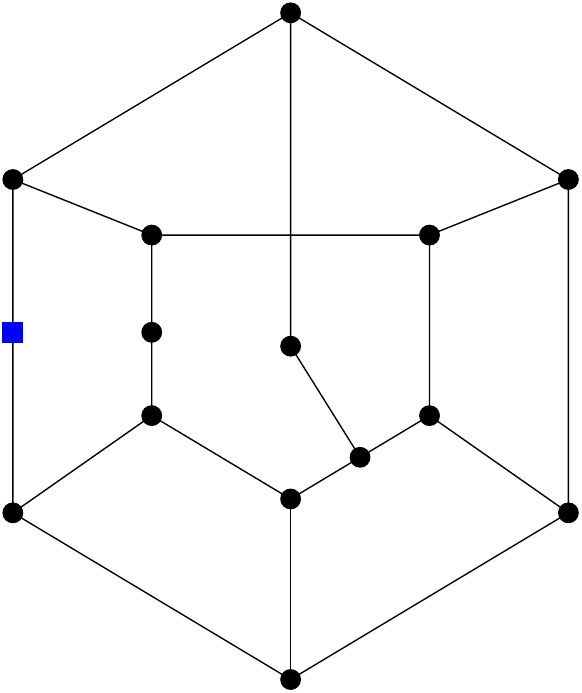}
    \end{minipage}
    \hfil\hfil
    \begin{minipage}{.45\textwidth}
      \center
      \small $\Z_4$: YES \hfil $\Z_2^2$: NO \\
      \includegraphics[width=.95\textwidth]{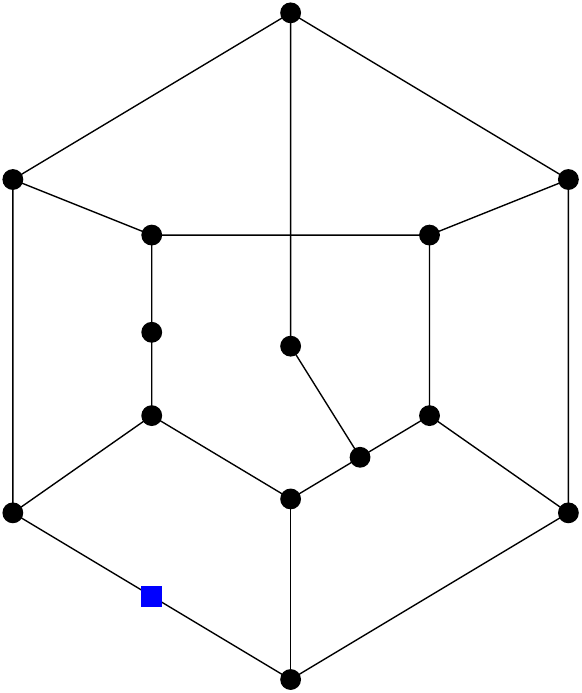}
    \end{minipage}
  \end{center}
  \caption{Graphs proving Theorem~\ref{main} \label{fig:graphs}}
\end{figure}

\FloatBarrier
\section{Group connectivity testing} \label{sec:algo}

We fix a digraph~$G=(V,E)$. We let $n$ be the number of vertices and 
$m$ the number of edges of~$G$.

\begin{notation}
We say that a flow $\varphi\colon E \to \Gamma$ {\em satisfies} a mapping of forbidden values
$h\colon E \to \Gamma$ if for every $e \in E$ it holds $h(e) \neq \varphi(e)$.
\end{notation}

The most straightforward way of testing whether a graph is $\Gamma$-connected,
is using the definition: We can enumerate all $h\colon E \to \Gamma$ assignments of
forbidden values and for each of them (try to) find a satisfying flow. Finding
a satisfying flow by itself is a hard problem: A cubic graph has nowhere-zero
$\Z_4$-flow (equivalently, $\Z_2^2$-flow) if and only if it has an edge 3-coloring.
Testing the edge 3-colorability of cubic graphs was shown to be NP-complete by Holyer~\cite{Ho81}.

An easy observation about the structure of forbidden assignments is:

\begin{obs}\label{obs:forb_assgn}
Let $h, h'\colon E \to \Gamma$ be assignments of the forbidden values such that
$h' - h = \Delta$ is a flow. Then $h$ is satisfied by a flow $\varphi$ if and only
if $h'$ is satisfied by $\varphi + \Delta$.
\end{obs}

\begin{dfn}
We say that assignments of forbidden values $h, h'\colon E \to \Gamma$
are flow-equivalent, denoted $h \sim_f h'$, if and only if $h' - h$ is a flow.
\end{dfn}

Hence we can split all assignments of the forbidden values into equivalence classes
of $\sim_f$ and test existence of a satisfying flow only for one member of each class.
This improves algorithm from finding
$\abs{\Gamma}^{m}$ flows to finding $\abs{\Gamma}^{n - 1}$ flows
(because every equivalence class is uniquely determined by an assignment of forbidden values
which is 0 outside of some fixed spanning tree).

A bit smarter algorithm -- used to find $\Z_2^2$-connected graphs which are not
$\Z_4$-connected -- can be obtained by looking
at Observation~\ref{obs:forb_assgn} the other way around. It follows that each
equivalence class of $\sim_f$ is exactly a coset generated by adding some its fixed member to all flows.
Therefore if an equivalence class $[x]_{\sim_f}$ is satisfied then for every flow $\varphi$ there is $h \in [x]_{\sim_f}$
such that $\varphi$ satisfies $h$.

\begin{theorem}
Fix a digraph $G$ and an abelian group $\Gamma$. Let $x\colon E \to \Gamma$ be
a forbidden mapping. The following statements are equivalent:
\begin{enumerate}
  \item Forbidden mapping $x$ is satisfied.
  \item Every $y \in [x]_{\sim_f}$ is satisfied.
  \item For every flow $\varphi$, there exists $y \in [x]_{\sim_f}$ satisfied by $\varphi$.
\end{enumerate}
\end{theorem}

\begin{proof}
Equivalence of first two is Observation~\ref{obs:forb_assgn}. For item three we fix a flow $\varphi_x$
satisfying $x$. Then flow $\varphi$ satisfies forbidden mapping $x - \varphi_x + \varphi$. And vice versa
if $\varphi$ satisfies $y$ then $x$ is satisfied by $\varphi - y + x$.
\end{proof}

So we can fix a flow -- constant-zero flow being the obvious candidate -- and
for each equivalence class we test whether some of its members is satisfied by it. This
increases the number of tests back to $\abs{\Gamma}^{m}$ but now each test is
just a simple comparison instead of an NP-complete problem.

We can also trade some
space for time: We keep a table of all equivalence classes, and instead of enumerating members
of all equivalence classes, we enumerate all assignments of forbidden values that are satisfied
by the given flow. For each of them we determine its equivalence class and mark that class as satisfied.
After enumerating them all we just check whether every equivalence class is satisfied.
This decreases the number of enumerated elements to $(\abs{\Gamma} - 1)^m$ but
consumes extra $2^{n - 1}$ bits of memory.

Because we were testing subdivisions of cubic graphs we would like to optimize
cases of once and twice subdivided edges. Without any additional optimization each
subdivision of an edge increases the number of edges by one and hence slows down
the described method by the factor of $\abs{\Gamma} - 1$. But
a subdivision creates an edge 2-cut.

Without loss of generality we may assume that edges of a 2-cut -- denote them $e_1$ and
$e_2$ -- are oriented in opposite
directions. The value of any flow must be the same on both of them. Hence swapping the forbidden
values for edges $e_1$ and $e_2$ does not change the set of satisfying flows. Moreover,
we may assume that forbidden values for $e_1$ and $e_2$ are different because it
is more restrictive than the case when they are the same. This reduces the number of
cases from $\abs{\Gamma}^2$ to ${\abs{\Gamma} \choose 2}$ (\ie from 16 to 6 for
groups of order four).
Double subdivision is in our case even simpler because we have
three forbidden values and again the most restrictive case is when they all are distinct.
So such double-subdivided edge has only one possible value (in our case, where $|\Gamma|=4$). 

Now we need to plug
this observations into above-described algorithm.
Observe that the equivalence classes used in the algorithm do not have
to be equivalence classes of $\sim_f$
but we can use classes of any equivalence $\sim$ which is congruence with respect
to satisfiability and which is coarser than $\sim_f$. Being congruence with
respect to satisfiability means that either all elements of equivalence class
are satisfiable or none of them is. Being coarser than $\sim_f$ ensures
that $[x]_{\sim_f} \subseteq [x]_{\sim}$ and so if class $[x]_{\sim}$ is
satisfiable that for every flow $\varphi$ there is some $y \in [x]_{\sim}$
satisfied by $\varphi$.
Moreover, we can throw away equivalence classes that are satisfied if some other class is satisfied
(of course without creating cycles).
E.g.~if we have 2-cut with both forbidden values being $1$, then this case is implied
by the case with value $1$ and any other value.

\begin{notation}
We let $[A \to B]$ denote the set of all functions from $A$ to $B$.
\end{notation}

We summarize our approach in Algorithm~\ref{alg} and Theorem~\ref{thm:alg}.
We also need to work with equivalence classes in the algorithm, so we represent the equivalence
with throw-away class as a function $$\mathcal C\colon [E \to \Gamma] \to X \uplus \set{\NULL}$$
which assigns to each forbidden mapping an object representing its class (in practical
implementation elements of $X$ are just small integers), $\NULL$ representing the throw-away class.

Function $\mathcal C$ we used is obtained from $\sim_f$ by following modifications:
For each 2-cut we remove all classes (\ie we set values of their elements to \NULL) that forbid the same value
on both edges of the cut and merge classes which differ only by swapping
values on edges of the cut. For double-subdivided edges we remove all classes that
do not forbid three different values on each double-subdivided edge and than merge all classes
that differ only by the order of forbidden values on given subdivided edge. We note that the optimization
for double-subdivided edges is essentially equivalent to removing given subdivided edge:

\begin{obs}
  If graph $G$ contains an edge subdivided $\abs{\Gamma}$ times, it cannot be $\Gamma$-connected.
  If it contains an edge $e$ subdivided $\abs{\Gamma}-1$ times, it is $\Gamma$-connected if and only if
  $G - e$ is $\Gamma$-connected.
\end{obs}

\def\true{\ensuremath{\text{\tt true}}}
\def\false{\ensuremath{\text{\tt false}}}

\begin{algorithm}
  \KwIn{graph $G$}
  \KwOut{{\tt YES} if $G$ is $\Gamma$-connected, and {\tt NO} otherwise}
  \BlankLine

  Pick a flow $\varphi_0$\;
  Create array $a$ indexed by elements of $X$\;
  $a[*] \leftarrow \false$\;
  \BlankLine

  \For{$\forall h$ {\rm satisfied by} $\varphi_0$ {\rm such that} $\mathcal C(h) \neq \NULL$}{
    $a[\mathcal C(h)] \leftarrow \true$\;
  }
  \BlankLine

  \For{$\forall x \in X$}{
    \lIf{$a[x] = \false$}{\Return{{\tt NO}}}
  }
  \BlankLine

  \Return{{\tt YES}}\;

  \caption{Group connectivity testing \label{alg}}
\end{algorithm}

\begin{theorem} \label{thm:alg}
  Fix an abelian group $\Gamma$, a digraph $G$, and a function
  $\mathcal C\colon [E \to \Gamma] \to X \uplus \set{\NULL}$ such that:
  \begin{enumerate}
    \item $X \subseteq \mathcal C [[E \to \Gamma]]$, \label{item:dom}
    \item for all $h\colon E \to \Gamma$ if $\mathcal C(h) = \NULL$ then there exits $h'\colon E \to \Gamma$
      such that if $h'$ is satisfied then $h$ is also satisfied
      and $\mathcal C(h') \neq \NULL$, \label{item:almost_equiv}
    \item for all $h, h'\colon E \to \Gamma$ if $\mathcal C(h) = \mathcal C(h')$ then either both are satisfied
      or none of them is, and \label{item:congruence}
    \item for all $h\colon E \to \Gamma$ and for all $\Gamma$-flows $\varphi$ holds
      $\mathcal C(h) = \mathcal C(h + \varphi)$.
      \label{item:coarser}
  \end{enumerate}

  Then Algorithm~\ref{alg} correctly decides whether $G$ is $\Gamma$-connected.
\end{theorem}

\begin{proof}
Obviously, Algorithm~\ref{alg} terminates.

First we prove that if the graph is $\Gamma$-connected, then the
algorithm outputs {\tt YES}. By contradiction, let $x \in X$ be the element that forced algorithm to output
{\tt NO}. Let $P = \mathcal C^{-1}(x)$ be set of preimages of $x$. It is nonempty due to Assumption~\ref{item:dom},
so we can fix some $p \in P$. The mapping $p$ is satisfied by some flow $\varphi_p$ because
$G$ is $\Gamma$-connected. The mapping $p' = p - \varphi_p + \varphi_0$ is satisfied by flow $\varphi_0$
(Observation~\ref{obs:forb_assgn}). Also $\mathcal C(p') = \mathcal C(p) = x$ (Assumption~\ref{item:coarser}),
so mapping $p'$ was enumerated by the algorithm and set $a[x]$ to \true. Contradiction.

Now we prove that if the algorithm outputs {\tt YES}, the graph $G$ is $\Gamma$-connected.
By contradiction, let $p\colon E \to \Gamma$ be a mapping witnessing that $G$ is not $\Gamma$-connected.
If $\mathcal C(p) = \NULL$, Assumption~\ref{item:almost_equiv} gives us $p'$ which is also unsatisfied and
$\mathcal C(p') \neq \NULL$, otherwise we take $p' = p$.
Because $\mathcal C(p') \neq \NULL$, none of the mappings in the set $\mathcal C^{-1}(\mathcal C(p'))$ is satisfied
(Assumption~\ref{item:congruence}). Hence $a[C(p')]$ was never set to $\true$, and the algorithm must have returned
{\tt NO}. Contradiction.
\end{proof}

\section{Implementation notes} \label{sec:impl}

Because large part of our work was creating programs for testing group connectivity, we would like
to add some implementation notes. Readers interested only in theoretical results may safely skip
this section. 

Our first implementation of straightforward algorithm
was written by the second author during her master thesis work. It was
a C++ implementation which was very specialized for the graphs tested (subdivisions of cube),
and a CSP implementation in Sicstus Prolog to double-check the results. Both of these implementations
required preprocessed input which made them less than ideal to work with, and also was not
fast enough for searching through larger graphs.

Hence we have written a new implementation based on Algorithm~\ref{alg}
in Python~2 built on Sage libraries which already contain
a lot of tools to work with general graphs.\footnote{
We used version 2 of Python because Sage was not yet ported to Python 3.
} Because Python is an interpreted language and as such is slower,
we chose to implement performance critical parts of code in C++ binding them into Python using
Cython.\footnote{ Do not mistake with CPython -- CPython is reference implementation
of Python interpreter, whereas Cython is optimizing compiler of Python which compiles
Python into C (or C++) and then into native code  using standard compiler like gcc.}

At the end of the previous section we have described function $\mathcal C$
that we are using, but we did not specify
how to calculate it. The main idea is to fix a spanning tree and transform any forbidden mapping
to an equivalent one which is zero outside this tree. To do so we keep a precalculated list of
elementary flows. We also need to take care of merged classes created by (doubly-)subdivided edges.
For doubly subdivided edges we always assign them the only interesting forbidden values
(and remove them from generation of forbidden mappings). For single subdivisions we keep list
of six interesting assignments and assign subdivided edges only values from this list.
Effect of these optimization is shown in Table~\ref{fig:timeit}.

\begin{table}
  \centering
  \caption{Time required to test cube subdivided on 2 edges (all 9 possibilities).}
  \label{fig:timeit}
  \smallskip
  \begin{tabular}{lc} \hline
    Algorithm & Time [s] \\ \hline
    Simple (in Python)  & 48.8 \\
    Smart (in C++)    & 3.65 \\
    Smart with subdivision optimization & 0.229 \\ \hline
  \end{tabular}
  \smallskip

  \small Measured on Intel i5 5257U.
\end{table}

To double-check our results we also implemented the straightforward algorithm in pure Python.  It is called Simple algorithm in Table~\ref{fig:timeit}. It
does just check the definition -- for every forbidden assignment (fixed outside of
a spanning tree) it finds a satisfying flow (from precomputed list of flows).
A repository with both implementations may be found at our department's GitLab
$$\text{\url{https://gitlab.kam.mff.cuni.cz/radek/group-connectivity-pub}.}$$

\section{Conclusions and open problems}

We have found graphs that show that $\Z_2^2$- and $\Z_4$-connectivity are independent notions. 
All of the graphs that we have found to certify this do have vertices of degree~2. Therefore, it is natural to ask, 
whether such graphs exist that are 3-edge-connected (both, in cubic and general case). 

Another challenging task is to find a proof that does not use computers. The main obstacle is to find 
efficient techniques to show that a particular graph is $\Gamma$-connected. To prove the converse 
is much easier: we guess forbidden values~$h\colon E \to \Gamma$ and then show non-existence of a flow
(see Theorem~\ref{thm:non-conn}). 

Our final question is the complexity of testing group connectivity. The algorithm we have 
developed is fast enough for our purposes; the required time is exponential, however. 
To test for group connectivity seems harder than to test for existence of a nowhere-zero flow, which suggests 
the problem is NP-hard. In fact, we believe it is $\Pi^p_2$-complete.

Circumstantial evidence which suggests
$\Pi^p_2$-completeness are somewhat dual notions of choosability and group list-colorings.
Both of these problems are known to be $\Pi^p_2$-complete -- proved by Erd\H{o}s et al.~\cite{ERT}
for choosability, and Kráľ~\cite{Dan05,Dan04} for group list-colorings.
Of those two, group list-colorings are closer match to dual of group connectivity, but graphs used in
Kráľ's proofs are non-planar, and we found no way to work around it.
So for testing group connectivity we do not know any hardness results.

\section{Acknowledgements} 
The first and the last author were partially supported by GAČR grant 16-19910S.
The first author was partially supported by the Charles University, project GA UK No.~926416.

\bibliographystyle{amsplain}
\bibliography{group_con}

\end{document}